\documentclass[a4paper,11pt,oneside]{article}
\usepackage[english]{babel}
\usepackage[dvips]{graphicx}
\usepackage[svgnames]{xcolor}

\colorlet{fixed}{LightGoldenrodYellow!75!White}
\colorlet{iter1}{Salmon!30!White}
\colorlet{iter2}{Salmon!55!White}

\colorlet{fixeda}{LightGreen!50!White}
\colorlet{iter1a}{SlateBlue!25!White}
\colorlet{iter2a}{SlateBlue!40!White}

\colorlet{halfgray}{DarkGrey}

\usepackage{pst-all} 
\newpsobject{showgrid}{psgrid}{subgriddiv=1,griddots=10,gridlabels=6pt}
\psset{arrowscale=2}

\usepackage{amssymb}
\usepackage[tbtags]{amsmath}
\usepackage{amsthm}
\usepackage{icomma}

\theoremstyle{definition}
 
\newtheorem{example}{Example}

\theoremstyle{plain}

\newtheorem{lemma}{Lemma}

\usepackage[T1]{fontenc}
\usepackage{lmodern}

\usepackage[top=32mm]{geometry}

\usepackage{caption}
\usepackage{booktabs}
\usepackage[flushleft]{threeparttable}

\usepackage{tabularx}

\usepackage[authoryear,round,longnamesfirst]{natbib}
\usepackage{fancyhdr}
\usepackage{hyperref}
\hypersetup{
  linktocpage=true,
  linkcolor  = MidnightBlue,
  citecolor  = MidnightBlue,
  urlcolor   = MidnightBlue,
  colorlinks = true,
}

\title{Pumping lemmas for classes of languages generated by folding systems\thanks{This is the accepted version of a manuscript published in {\it Natural Computing} (2021), 20:321--327, https://doi.org/10.1007/s11047-019-09771-5, \copyright{} Springer Nature B.V. 2019.
}}
\author{Jorge C. Lucero\thanks{Dept.\ Computer Science, University of Bras\'{i}lia, Brazil. E-mail: \href{mailto:lucero@unb.br}{lucero@unb.br} }} 
\date{November 25, 2026} 

\usepackage[euler]{textgreek}

\newcommand{\str}[1]{\texttt{#1}}
\newcommand{\eee}{\texttt{\textepsilon}}
\DeclareMathOperator{\lcm}{lcm}

\begin{document}

\maketitle

\begin{abstract}
Geometric folding processes are ubiquitous in natural systems ranging from protein biochemistry to patterns of insect wings and leaves. In a previous study, a folding operation between strings of formal languages was introduced as a model of such processes. The operation was then used to define a folding system (F-system) as a construct consisting of a core language, containing the strings to be folded, and a folding procedure language, which defines how the folding is done. This paper reviews main definitions associated with F-systems and next it determines necessary conditions for a language to belong to classes generated by such systems. The conditions are stated in the form of pumping lemmas and four classes are considered, in which the core and folding procedure languages are both regular, one of them is regular and the other context-free, or both are context-free. Full demonstrations of the lemmas are provided, and the analysis is illustrated with examples. 
\end{abstract}

\section{Introduction}
\label{intro}

In a recent paper, \citet{Sburlan2011} introduced a folding operation for strings of symbols of a given formal language. The operation was inspired in actual geometric folding processes that occur in, e.g., protein biochemistry \citep{Dobson2003}, in-vitro DNA shaping \citep{Rothemund2006}, and even origami \citep[the Japanese art of paper folding;\linebreak][]{Demaine2007},
and it was proposed as a restricted computational model of such processes. 

The folding operation was applied to define folding systems (F-systems) of the form $\Phi=(L_1, L_2)$, where $L_1$ is the language that contains the strings to be folded (the core language), and $L_2$ is the language that contains strings defining how the folding must be performed (the folding procedure language). Then, the computing power of F-systems was investigated by comparison with standard language classes from the Chomsky hierarchy; i.e., regular, context-free, context-sensitive, recursive and recursively enumerable languages. 

The paper fitted well within a growing body of applications of geometric folding to science and technology which have surged in recent years; in, e.g., aerospace and automotive technology \citep{Cipra2001}, civil engineering \citep{Filipov2015}, biology \citep{Mahadevan2005}, and robotics \citep{Felton2014}. Also, a number of theoretical studies have considered algebraic models, algorithmic complexity, and other mathematical and computational aspects of folding\citep[e.g.,][]{Akitaya2017,Alperin2000, Ida2016}. 

F-systems might find relevant applications to DNA computing and related areas of natural computing \citep{Kari2008}. Thus, the general purpose of the present paper is to further explore capabilities and limitations of such systems, and it will consider the classes of languages generated when the core and the folding procedure languages are regular or context-free. Necessary conditions for a language to belong to some of those classes were presented by \citet{Sburlan2011} in the form of pumping lemmas, similar to the well known pumping lemmas
for regular and context-free languages \citep{Hopcroft2001}.  His analysis considered the cases in which both the core and the folding procedure language are regular, and in which the core language is context-free and the folding procedure language is regular. Here,  the two cases  will be revised, in order to solve detected inconsistencies  (see footnote \ref{fsbu} to Lemma \ref{lemma1}). 
The lemmas will be restated in a weaker form (as consequence of the revision) and full proofs will be provided. Also, it will be shown that the same lemma for the case in which the core language is context-free and the folding procedures language is regular also applies to the case in which the class attribution is reversed. Finally, a lemma for the remaining case in which both the core and the folding procedure languages are context-free will be presented and proved. 

\section{Folding systems}
\label{fsystems}

For clarity of the analysis, let us review main definitions associated to folding operations and systems. 

Let $\Sigma$ be an alphabet, $\Gamma=\{\str{u}, \str{d}\}$, and $f:\Sigma^*\times\Sigma\times\Gamma\rightarrow\Sigma^*$ a function such that
\begin{equation}
f(w, a, b)= \left\{
\begin{array}{ll}
aw, & \text{ if } b=\str{u},\\
wa, & \text{ if } b=\str{d}.
\end{array}
\right.
\end{equation}
Then, the folding function $h:\Sigma^*\times\Gamma^*\rightarrow\Sigma^*$ is a partial function defined by
\begin{equation}
\label{deff}
h(w, v)=\left\{
\begin{array}{ll}
f(f(\ldots f(\eee, a_1, b_1)\ldots, a_{k-1}, b_{k-1}), a_k, b_k), &\text{ if } |w|=|v|>0\\
\eee, & \text{ if } |w|=|v|=0,\\
\text{undefined,} & \text{ if } |w|\neq|v|.
\end{array}
\right.
\end{equation}
where $w=a_1a_2\ldots a_k$, $v=b_1b_2\ldots b_k$, $a_i\in\Sigma$ and $b_i\in\Gamma$ for $i=1, 2, \cdots, k$, and $\eee$ is the empty string.\footnote{\citeauthor{Sburlan2011}'s \citeyearpar{Sburlan2011} original definition has been modified in order to include the case in which both $w$ and $v$ are empty strings.} 

\begin{example}
\label{exmp1}
Let $w=\str{abcde}$ and $v=\str{dduud}$. Then,
\begin{align*}
h(w, v)&=f(f(f(f(f(\eee, \str{a}, \str{d}), \str{b}, \str{d}), \str{c}, \str{u}), \str{d}, \str{u}), \str{e}, \str{d}),\\
       &=f(f(f(f(\str{a}, \str{b}, \str{d}), \str{c}, \str{u}), \str{d}, \str{u}), \str{e}, \str{d}),\\
       &=f(f(f(\str{ab}, \str{c}, \str{u}), \str{d}, \str{u}), \str{e}, \str{d}),\\  
       &=f(f(\str{cab}, \str{d}, \str{u}), \str{e}, \str{d}),\\  
       &=f(\str{dcab}, \str{e}, \str{d}),\\  
       &=\str{dcabe}.  
\end{align*}

The computation of $h(w, v)$ is represented graphically in Fig.\ \ref{folding}. As shown there, each application of function $f$ may be regarded as a folding operation that arranges the symbols of $w$ in a stack. Strings over $\Gamma$ describe how each folding must be performed, where symbol \str{u} represents a ``folding up'' action and symbol \str{d} represents a ``folding down'' action. The final result is the created stack, read from top to bottom.

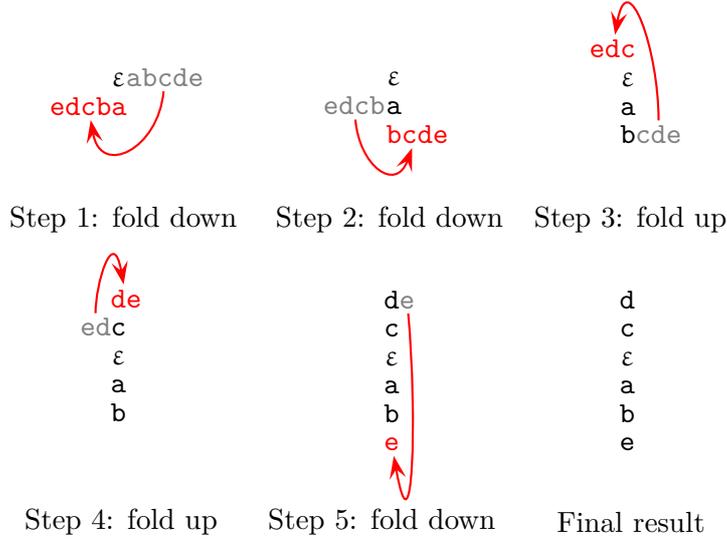
\begin{figure}
\centering
\begin{pspicture}(0,1)(10,8)
\uput[0](0.4,6.7){
\begin{psmatrix}[colsep=0cm,rowsep=-.1cm]
\str{\ \ \ \ }\eee & \rnode{n1}{\textcolor{Gray}{\str{abcde}}}\\
\rnode{n2}{\textcolor{Red}{\str{edcba}}}
\end{psmatrix}
\nccurve[nodesep=2pt,linecolor=Red,angleA=-95, angleB=-80, ncurv=1.5]{->}{n1}{n2}
}
\uput[0](0,5.){Step 1: fold down}

\uput[0](4.,6.49){
\begin{psmatrix}[colsep=0.02cm,rowsep=-.1cm]
&\eee\str{\ \ \ \ }\\
\rnode{n2}{\textcolor{Gray}{\str{edcb}}}& \str{a\ \ \ \ }\\
& \rnode{n3}{\textcolor{Red}{\str{bcde}}}
\end{psmatrix}
\nccurve[nodesep=2pt,linecolor=Red,angleA=-85, angleB=-110, ncurv=1.5]{->}{n2}{n3}
}
\uput[0](3.5,5.){Step 2: fold down}

\uput[0](7.5,6.7){
\begin{psmatrix}[colsep=0cm,rowsep=-.1cm]
\rnode{n4}{\textcolor{Red}{\str{edc}}}\\
\str{\ \ }\eee\\
\str{\ \ a}\\
\str{\ \ b} & \rnode{n3}{\textcolor{Gray}{\str{cde}}}
\end{psmatrix}
\nccurve[nodesep=2pt,linecolor=Red,angleA=90, angleB=80, ncurv=1.5]{->}{n3}{n4}
}
\uput[0](6.9,5.){Step 3: fold up}

\uput[0](.8,3.2){
\begin{psmatrix}[colsep=0cm,rowsep=-.1cm]
                                  &\rnode{n5}{\textcolor{Red}{\str{de}}}\\
\rnode{n4}{\textcolor{Gray}{\str{ed}}}&\str{c\ \ }\\
                                  &\eee\str{\ \ }\\
                                  &\str{a\ \ }\\
                                  &\str{b\ }
\end{psmatrix}
\nccurve[nodesep=2pt,linecolor=Red,angleA=90, angleB=100, ncurvA=1.5, ncurvB=3]{->}{n4}{n5}
}
\uput[0](0.2,1){Step 4: fold up}

\uput[0](4.8,3){
\begin{psmatrix}[colsep=0cm,rowsep=-.1cm]
\str{d} & \rnode{n5}{\textcolor{Gray}{\str{e}}}\\
\str{c}\\
\eee\\
\str{a}\\
\str{b}\\
\rnode{n6}{\textcolor{Red}{\str{e}}}																	
\end{psmatrix}
\nccurve[nodesep=2pt,linecolor=Red,angleA=-85, angleB=-80, ncurv=1.5]{->}{n5}{n6}
}
\uput[0](3.4,1){Step 5: fold down}
\uput[0](7.9,3){
\begin{psmatrix}[colsep=0cm,rowsep=-.1cm]
\str{d}\\
\str{c}\\
\eee\\
\str{a}\\
\str{b}\\
\str{e}																	
\end{psmatrix}
}
\uput[0](7.2,1.){Final result}
\end{pspicture} 

\caption{Computation of $h(\str{abcde}, \str{dduud})=\str{dcabe}$ as a sequence of folding operations.} 
\label{folding}
\end{figure}

\end{example}

Another way to see the folding operation is illustrated in Fig.~\ref{folding2}. The folded string may be written as
\begin{equation}
h(w, v) = w_\str{u}^\mathcal{R}w_\str{d},
\label{2way}
\end{equation} 
where $w_\str{u}$ is the sequence of symbols in $w$ that are folded up, $\mathcal{R}$ denotes the reverse order operator, and $w_\str{d}$ is the sequence of symbols in $w$ that are folded down. Letting $w=xy$, $v=st$, with $|x|=|s|$ and $|y|=|t|$, we obtain the identity
\begin{equation}
h(xy, st)  = y_\str{u}^\mathcal{R}h(x, s)y_\str{d}.
\label{prop1}
\end{equation} 

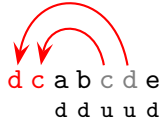
\begin{figure}
\centering
\begin{pspicture}(5,1.5)(7.5,4)
\uput[0](5,2){
\begin{psmatrix}[colsep=.1cm,rowsep=-.1cm]
\rnode{n1}{\textcolor{Red}{\str{d}}} & \rnode{n2}{\textcolor{Red}{\str{c}}}
& \str{a} & \str{b} & \rnode{n3}{\textcolor{Gray}{\str{c}}} & \rnode{n4}{\textcolor{Gray}{\str{d}}} & \str{e}\\
&&\footnotesize\str{d}&\footnotesize\str{d}&\footnotesize\str{u}&\footnotesize\str{u}&\footnotesize\str{d}
\end{psmatrix}
}
\nccurve[nodesep=2pt,linecolor=Red,angleA=100, angleB=80, ncurv=1.5]{->}{n4}{n1}
\nccurve[nodesep=4pt,linecolor=Red,angleA=100, angleB=80, ncurv=1.5]{->}{n3}{n2}
\end{pspicture} 
\caption{Computation of $h(\str{abcde}, \str{dduud})=\str{dcabe}$ using Eq.~(\ref{2way}).} 
\label{folding2}
\end{figure}

A folding system (F-system) is defined as a pair $\Phi=(L_1, L_2)$, where $L_1\subseteq \Sigma^*$ is the core language, and $L_2\subseteq \Gamma^*$ is the folding procedure language. The language of $\Phi$ is
\begin{equation}
L(\Phi)=\{h(w, v)|w\in L_1, v\in L_2, |w|=|v|\}. 
\end{equation}

The class of all languages generated by F-systems with core languages of a class $\mathcal{C}$ and folding procedure languages of a class $\mathcal{H}$ is defined as
\begin{equation}
\mathcal{F}(\mathcal{C}, \mathcal{H})=\{L(\Phi)|\Phi=(L_1, L_2), L_1\in\mathcal{C}, L_2\in\mathcal{H}\}.
\end{equation}

\section{Pumping lemmas}
\label{pumping}

The first lemma states necessary conditions for a language to belong to class \linebreak $\mathcal{F}(\text{\sffamily REG}, \text{\sffamily REG})$, where  \text{\sffamily REG} is the class of regular languages.

\begin{lemma}
Let $L \in \mathcal{F}(\text{\sffamily REG}, \text{\sffamily REG})$ be an infinite language over an alphabet $\Sigma$. Then, there are strings $u, v, x, y, z \in \Sigma^*$, with $|vy|>0$, such that $uv^ixy^iz\in L$ for all $i\ge 0$.\footnote{\citeauthor{Sburlan2011}'s \citeyearpar{Sburlan2011} original version of the lemma states that there exists a positive constant $p$ such that any string $w\in L$, with $|w| \ge p$, can be written as $w =uvxyz$ satisfying $uv^ixy^iz\in L$ for each $i\ge 0$, $|vy| \ge 1$, and $|vxy|\le p$. In his proof, he set $p=\max(p_1, p_2)$ where $p_1$ and $p_2$ are the pumping lengths for the core and the folding procedure languages. However, it can be shown that such a value of $p$ does not always work. As a simple example, set $\Phi=(L_1, L_2)$ with $L_1=\str{aaaab}^*$ and $L_2=(\str{uu})^*\str{ddd}$. Then, $p=5$. However, $L(\Phi)= \str{aaaab} \cup (\str{bb})^*\str{aaaabbb}$, and note that string \str{aaaab} has length 5 but it can not be pumped as indicated by the lemma. Although the original lemma may still hold for some other value of $p$ ($p=9$, in case of the example, with $u=x=y=\eee$, $v=\str{bb}$ and $z=\str{aaaabbb}$), a full demonstration was not provided. 

Other inconsistencies have been detected in the proof. For example, the proof is based on constructing a double stranded structure $\left[\frac{r_j}{s_j}\right]$, and the technique is also used here. However, instead of Eqs.~(\ref{rj}) and (\ref{sj}), the previous proof sets $r_j = x_ry_r^{\frac{\lcm(|y_r|, |y_s|)}{|y_r|}j}z_r$ and
$s_j = x_sy_s^{\frac{\lcm(|y_r|, |y_s|)}{|y_s|}j}z_s$, where $\lcm$ denotes the least common multiple. Since the stranded construction requires $|r_j|=|s_j|$, it follows that $|x_rz_r|=|x_sz_s|$ and hence $|y_r|=|y_s|$. However, those conditions do not hold for arbitrary regular languages $L_1$ and $L_2$ (note that, in the above example, $|y_r|=|\str{b}|=1$ whereas $|y_s|=|\str{uu}|=2$).

Similar problems have been found also in the original version of the lemma for the case of $L\in\mathcal{F}(\text{\sffamily CF}, \text{\sffamily REG})$.\label{fsbu}}
\label{lemma1}  
\end{lemma}
\begin{proof}
Let $L_1\subseteq\Sigma^*$, $L_2\subseteq\{\str{u}, \str{d}\}^*$ and define the F-system $\Phi=(L_1, L_2)$. If $L=L(\Phi)$ is infinite, then $L_1$ and $L_2$ are also infinite. Set $p=\max(p_1, p_2)$, where $p_1$ and $p_2$ are pumping lengths for $L_1$ and $L_2$ (from the pumping lemma for regular languages), respectively, and choose any string $w\in L$ such that $|w|\ge p$. Then, there are strings $r\in L_1$ and $s\in L_2$ such that $w=h(r, s)$, with $|r|=|s|=|w|\ge p$. 

According to the pumping lemma for regular languages, $r$ may be written as $r=x_ry_rz_r$, with  $|y_r|>0$, such that  $x_ry_r^jz_r\in L_1$ for any $j\ge 0$. Similarly, $s$ may be written as $s=x_sy_sz_s$, with $|y_s|>0$, such that $x_sy_s^jz_s\in L_2$ for any $j\ge0$.

Consider the sequences $\{r_j\}$ and $\{s_j\}$ defined by 
\begin{align}
r_j& = x_ry_r^{|y_s|j+1}z_r,\label{rj}\\
s_j& = x_sy_s^{|y_r|j+1}z_s,\label{sj}
\end{align}
for any $j\ge 0$. Clearly, $r_j\in L_1$, $s_j\in L_2$; further, $|r_j|=|s_j|$ and so each of the strings in $\{r_j\}$ may be folded according with the corresponding string in $\{s_j\}$. For clarity, let us arrange $r_j$ and $s_j$ in the form of a double stranded structure $\left[\frac{r_j}{s_j}\right]$, illustrated with an example in Fig.~\ref{strand}. Note that, for $j$ large enough, repetitions of substring $y_r$ overlap with repetitions of substring $y_s$. Therefore, we may seek expressions of $r_j$ and $s_j$ of the form   
\begin{align}
r_j &=\xi_1\xi_2^{j-j_0}\xi_3, \label{ri1}\\
s_j &=\mu_1\mu_2^{j-j_0}\mu_3, \label{si1}
\end{align}
with substrings $\xi_k$ and $\mu_k$ such that $|\xi_k|=|\mu_k|$ for $k=1, 2, 3$, and for $j\ge j_0$, where $j_0$ is a constant to be determined.

\begin{figure}
\centering
\begin{pspicture}(-2,1)(11.5,4.5)
\uput[0](-1.75,2.67){
$\begin{bmatrix}
r_j\\
s_j
\end{bmatrix}
=
\begin{bmatrix}
\hspace*{11.1cm}\\
\ 
\end{bmatrix}
$
}

\psframe[linecolor=fixed, fillstyle=solid,fillcolor=fixed](0,2.7)(10.9,3.15)
\psframe[linecolor=iter1, fillstyle=solid,fillcolor=iter1](1.3,2.7)(9.5,3.15)
\psframe[linecolor=iter2, fillstyle=solid,fillcolor=iter2](1.7,2.7)(8.3,3.15)

\psframe[linecolor=fixeda, fillstyle=solid,fillcolor=fixeda](0,2.15)(10.9,2.55)
\psframe[linecolor=iter1a, fillstyle=solid,fillcolor=iter1a](1.7,2.15)(9.9,2.55)
\psframe[linecolor=iter2a, fillstyle=solid,fillcolor=iter2a](1.7,2.15)(8.3,2.55)

\psline[linewidth=.5pt](-1.28,2.65)(-.95,2.65)
\psline[linewidth=.5pt](0,2.65)(10.9,2.65)
\uput[90](0,2.55){$|$}
\uput[90](.65,2.6){$x_r$}
\uput[90](1.3,2.55){$|$}
\uput[90](1.7,2.6){$y_r$}
\uput[90](2.1,2.55){$|$}
\uput[90](2.5,2.6){$y_r$}
\uput[90](2.9,2.55){$|$}
\uput[90](3.3,2.6){$y_r$}
\uput[90](3.7,2.55){$|$}
\uput[90](4.1,2.6){$y_r$}
\uput[90](4.5,2.55){$|$}
\uput[90](4.9,2.6){$y_r$}
\uput[90](5.3,2.55){$|$}

\uput[90](6.3,2.55){$|$}
\uput[90](6.7,2.6){$y_r$}
\uput[90](7.1,2.55){$|$}
\uput[90](7.5,2.6){$y_r$}
\uput[90](7.9,2.55){$|$}
\uput[90](8.3,2.6){$y_r$}
\uput[90](8.7,2.55){$|$}
\uput[90](9.1,2.6){$y_r$}
\uput[90](9.5,2.55){$|$}
\uput[90](10.2,2.6){$z_r$}
\uput[90](10.9,2.55){$|$}

\uput[90](0,2){$|$}
\uput[90](.85,2.05){$x_s$}
\uput[90](1.7,2){$|$}
\uput[90](1.97,2.05){$y_s$}
\uput[90](2.23,2){$|$}
\uput[90](2.5,2.05){$y_s$}
\uput[90](2.77,2){$|$}
\uput[90](3.03,2.05){$y_s$}
\uput[90](3.3,2){$|$}
\uput[90](3.57,2.05){$y_s$}
\uput[90](3.83,2){$|$}
\uput[90](4.1,2.05){$y_s$}
\uput[90](4.37,2){$|$}
\uput[90](4.63,2.05){$y_s$}
\uput[90](4.9,2){$|$}

\uput[90](5.8,2.1){$\cdots$}
\uput[90](5.8,2.6){$\cdots$}
\uput[90](5.8,1.1){$\cdots$}
\uput[90](5.8,3.7){$\cdots$}

\uput[90](6.7,2){$|$}
\uput[90](6.97,2.05){$y_s$}
\uput[90](7.23,2){$|$}
\uput[90](7.5,2.05){$y_s$}
\uput[90](7.76,2){$|$}
\uput[90](8.03,2.05){$y_s$}
\uput[90](8.3,2){$|$}
\uput[90](8.57,2.05){$y_s$}
\uput[90](8.83,2){$|$}
\uput[90](9.1,2.05){$y_s$}
\uput[90](9.37,2){$|$}
\uput[90](9.63,2.05){$y_s$}
\uput[90](9.9,2){$|$}
\uput[90](10.4,2.05){$z_s$}
\uput[90](10.9,2){$|$}

\psbrace[fillcolor=gray,rot=-90,ref=bC](1.7,3.2)(0,3.2){$\xi_1$}
\psbrace[fillcolor=gray,rot=-90,ref=bC](3.3,3.2)(1.7,3.2){$\xi_2$}
\psbrace[fillcolor=gray,rot=-90,ref=bC](4.9,3.2)(3.3,3.2){$\xi_2$}
\psbrace[fillcolor=gray,rot=-90,ref=bC](8.3,3.2)(6.7,3.2){$\xi_2$}
\psbrace[fillcolor=gray,rot=-90,ref=bC](10.9,3.2)(8.3,3.2){$\xi_3$}

\psbrace[fillcolor=gray,rot=90,ref=tC](0,2.1)(1.7,2.1){$\mu_1$}
\psbrace[fillcolor=gray,rot=90,ref=tC](1.7,2.1)(3.3,2.1){$\mu_2$}
\psbrace[fillcolor=gray,rot=90,ref=tC](3.3,2.1)(4.9,2.1){$\mu_2$}
\psbrace[fillcolor=gray,rot=90,ref=tC](6.7,2.1)(8.3,2.1){$\mu_2$}
\psbrace[fillcolor=gray,rot=90,ref=tC](8.3,2.1)(10.9,2.1){$\mu_3$}

\end{pspicture}

\caption{Double stranded structure built with strings $r_j\in L_1$ and $s_j\in L_2$, in the case in which both $L_1$ and $L_2$ are regular languages.}
\label{strand}
\end{figure}
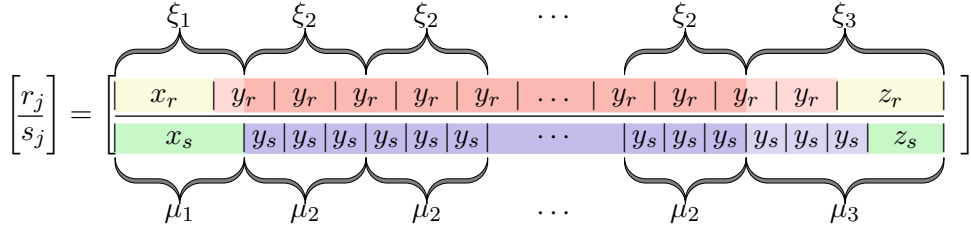

In Eqs.~(\ref{rj}) and (\ref{sj}), the iterated portions may be decomposed as
\begin{align}
 y_r^{|y_s|j+1} & =\alpha_1\xi_2^{j-j_0}\alpha_2,\\
 y_s^{|y_r|j+1} & =\beta_1\mu_2^{j-j_0}\beta_2,
\end{align}
with 
\begin{align}
\alpha_1\alpha_2 &= y_r^{j_0|y_s|+1},\label{aux1}\\
\beta_1\beta_2 &= y_s^{j_0|y_r|+1}.\label{aux2}
\end{align}
Consequently, $\xi_2$ and $\mu_2$ will be formed by repetitions of strings $y_r$ and $y_s$, respectively, with 
\begin{align}
|\xi_2^{j-j_0}| & = |y_r^{|y_s|j+1}| - |\alpha_1\alpha_2|,\nonumber \\
               & = |y_r|(|y_s|j+1) - |y_r|(j_0|y_s|+1) \nonumber\\
               & = |y_r||y_s|(j-j_0),
\end{align}
and the same result for $|\mu_2^{j-j0}|$. Consequently, $|\xi_2|=|\mu_2|=|y_r||y_s|$.

Next, we set $\xi_1 =x_r\alpha_1$, $\xi_3 =\alpha_2z_r$, $\mu_1 =x_s\beta_1$, $\mu_3 =\beta_2z_s$. 
Substrings $\alpha_1$, $\alpha_2$, $\beta_1$, $\beta_2$ are computed so that $|\xi_1|=|\mu_1|$ and $|\xi_3|=|\mu_3|$, with 
\begin{equation}
\left\{
\begin{array}{l}
\beta_1=\eee, |\alpha_1|= |x_s|-|x_r|,  \text{ if } |x_s|\ge|x_r|,\\
\alpha_1=\eee, |\beta_1|= |x_r|-|x_s|,  \text{ if } |x_s|<|x_r|.  
\end{array}
\right.\label{ab}
\end{equation}   
and using Eqs.\ (\ref{aux1}) and (\ref{aux2}). During the calculations, constant $j_0$ is chosen large enough so as to produce $|\alpha_2|\ge 0$ and $|\beta_2|\ge 0$.  
result.

Finally, $h(r_j, s_j)$ is computed by using Eqs.\ (\ref{ri1}), (\ref{si1}), and the identity in Eq.\ (\ref{prop1}). Note that $\xi_1$, $\xi_2$ and $\xi_3$ are folded as determined by $\mu_1$, $\mu_2$, and $\mu_3$, respectively. Then, letting $i=j-j_0$,
\begin{equation}
h(\xi_1\xi_2^i\xi_3, \mu_1\mu_2^i\mu_3)  = (\xi_3)_\str{u}^\mathcal{R}[(\xi_2)_\str{u}^\mathcal{R}]^i(\xi_1)_\str{u}^\mathcal{R}(\xi_1)_\str{d}(\xi_2)_\str{d}^i(\xi_3)_\str{d}.
\label{h1}
\end{equation}
The condition stated by the lemma is obtained from Eq.~(\ref{h1}) with $u=(\xi_3)_\str{u}^\mathcal{R}$, $v=(\xi_2)_\str{u}^\mathcal{R}$, $x=(\xi_1)_\str{u}^\mathcal{R}(\xi_1)_\str{d}$, $y=(\xi_2)_\str{d}$, $z=(\xi_3)_\str{d}$. Also, $|vy|=|\xi_2|=|y_r||y_s|>0$. 

\end{proof}

The next lemma considers the two cases in which the core and the folding procedure languages belong to different classes, regular or context-free. Its demonstration follows similar steps as in Lemma 1. 

\begin{lemma}
Let $L \in \mathcal{F}(\text{\sffamily  CF}, \text{\sffamily REG})$ or $L \in \mathcal{F}(\text{\sffamily  REG}, \text{\sffamily CF})$  be an infinite language over an alphabet $\Sigma$. Then, there are strings $w_1, w_2, \ldots, w_9\in \Sigma^*$, with  $|w_2w_4w_6w_8| >0$, such that $w_1w_2^iw_3w_4^iw_5w_6^iw_7w_8^iw_9\in L$ for all $i\ge 0$. 
\end{lemma}

\begin{proof}

Consider first the case $L \in \mathcal{F}(\text{\sffamily  CF}, \text{\sffamily REG})$ and proceed similarly as in Lemma 1, except that the pumping lemma for context-free languages is used for string $r\in L_1$. Thus, $r$ is written as $r=u_rv_rx_ry_rz_r$, with  $|v_ry_r|>0$ and $u_rv_r^jx_ry_r^jz_r\in L_1$, for any $j\ge 0$. Strings $r_j\in L_1$ and $s_j\in L_2$ are now defined as
\begin{align}
r_j& = u_rv_r^{|y_s|j+1}x_ry_r^{|y_s|j+1}z_r,\label{rj2}\\
s_j& = x_sy_s^{|v_ry_r|j+1}z_s,\label{sj2}
\end{align}
for any $j\ge 0$, and their double stranded arrangement is illustrated in Fig.~\ref{strand2}.  

\begin{figure}
\centering
\begin{pspicture}(-2,1)(12,4.5)
\uput[0](-1.75,2.67){
$\begin{bmatrix}
r_j\\
s_j
\end{bmatrix}
=
\begin{bmatrix}
\hspace*{11.8cm}\\
\ 
\end{bmatrix}
$
}

\psframe[linecolor=fixed, fillstyle=solid,fillcolor=fixed](0,2.7)(11.6,3.15)
\psframe[linecolor=iter1, fillstyle=solid,fillcolor=iter1](1,2.7)(5.2,3.15)
\psframe[linecolor=iter1, fillstyle=solid,fillcolor=iter1](5.7,2.7)(9.9,3.15)
\psframe[linecolor=iter2, fillstyle=solid,fillcolor=iter2](1.3,2.7)(4.9,3.15)
\psframe[linecolor=iter2, fillstyle=solid,fillcolor=iter2](6.3,2.7)(9.9,3.15)

\psframe[linecolor=fixeda, fillstyle=solid,fillcolor=fixeda](0,2.15)(11.6,2.55)
\psframe[linecolor=iter1a, fillstyle=solid,fillcolor=iter1a](1.3,2.15)(10.5,2.55)
\psframe[linecolor=iter2a, fillstyle=solid,fillcolor=iter2a](1.3,2.15)(4.9,2.55)
\psframe[linecolor=iter2a, fillstyle=solid,fillcolor=iter2a](6.3,2.15)(9.9,2.55)

\uput[90](0,2.55){$|$}
\uput[90](.5,2.6){$u_r$}
\uput[90](1,2.55){$|$}
\uput[90](1.3,2.6){$v_r$}
\uput[90](1.6,2.55){$|$}
\uput[90](1.9,2.6){$v_r$}
\uput[90](2.2,2.55){$|$}
\uput[90](2.5,2.6){$v_r$}
\uput[90](2.8,2.55){$|$}
\uput[90](3.1,2.6){$\cdots$}
\uput[90](3.1,3.7){$\cdots$}

\uput[90](3.4,2.55){$|$}
\uput[90](3.7,2.6){$v_r$}
\uput[90](4.0,2.55){$|$}
\uput[90](4.3,2.6){$v_r$}
\uput[90](4.6,2.55){$|$}
\uput[90](4.9,2.6){$v_r$}

\uput[90](5.2,2.55){$|$}
\uput[90](5.45,2.6){$x_r$}
\uput[90](5.7,2.55){$|$}

\uput[90](6.,2.6){$y_r$}
\uput[90](6.3,2.55){$|$}
\uput[90](6.6,2.6){$y_r$}
\uput[90](6.9,2.55){$|$}
\uput[90](7.2,2.6){$y_r$}
\uput[90](7.5,2.55){$|$}
\uput[90](8.1,2.6){$\cdots$}
\uput[90](8.1,3.7){$\cdots$}

\uput[90](8.7,2.55){$|$}
\uput[90](9.,2.6){$y_r$}
\uput[90](9.3,2.55){$|$}
\uput[90](9.6,2.6){$y_r$}
\uput[90](9.9,2.55){$|$}
\uput[90](10.8,2.6){$z_r$}
\uput[90](11.6,2.55){$|$}

\psline[linewidth=.5pt](-1.28,2.65)(-.95,2.65)
\psline[linewidth=.5pt](0,2.65)(11.6,2.65)

\uput[90](0,2){$|$}
\uput[90](0.65,2.05){$x_s$}
\uput[90](1.3,2){$|$}
\uput[90](1.5,2.05){\small $y_s$}
\uput[90](1.7,2){$|$}
\uput[90](1.9,2.05){\small $y_s$}
\uput[90](2.1,2){$|$}
\uput[90](2.3,2.05){\small $y_s$}
\uput[90](2.5,2){$|$}

\uput[90](3.1,1.1){$\cdots$}
\uput[90](3.1,2.1){$\cdots$}

\uput[90](3.7,2){$|$}
\uput[90](3.9,2.05){\small $y_s$}
\uput[90](4.1,2){$|$}
\uput[90](4.3,2.05){\small $y_s$}
\uput[90](4.5,2){$|$}
\uput[90](4.7,2.05){\small $y_s$}
\uput[90](4.9,2){$|$}
\uput[90](5.1,2.05){\small $y_s$}
\uput[90](5.3,2){$|$}
\uput[90](5.5,2.05){\small $y_s$}
\uput[90](5.7,2){$|$}
\uput[90](5.9,2.05){\small $y_s$}
\uput[90](6.1,2){$|$}
\uput[90](6.3,2.05){\small $y_s$}
\uput[90](6.5,2){$|$}
\uput[90](6.7,2.05){\small $y_s$}
\uput[90](6.9,2){$|$}
\uput[90](7.1,2.05){\small $y_s$}
\uput[90](7.3,2){$|$}
\uput[90](7.5,2.05){\small $y_s$}
\uput[90](7.7,2){$|$}

\uput[90](8.1,1.1){$\cdots$}
\uput[90](8.1,2.1){$\cdots$}

\uput[90](8.5,2){$|$}
\uput[90](8.7,2.05){\small $y_s$}
\uput[90](8.9,2){$|$}
\uput[90](9.1,2.05){\small $y_s$}
\uput[90](9.3,2){$|$}
\uput[90](9.5,2.05){\small $y_s$}
\uput[90](9.7,2){$|$}
\uput[90](9.9,2.05){\small $y_s$}
\uput[90](10.1,2){$|$}
\uput[90](10.3,2.05){\small $y_s$}
\uput[90](10.5,2){$|$}
\uput[90](10.85,2.05){$z_s$}
\uput[90](11.6,2){$|$}

\psbrace[fillcolor=gray,rot=-90,ref=bC](1.3,3.2)(0,3.2){$\xi_1$}
\psbrace[fillcolor=gray,rot=-90,ref=bC](2.5,3.2)(1.3,3.2){$\xi_2$}
\psbrace[fillcolor=gray,rot=-90,ref=bC](4.9,3.2)(3.7,3.2){$\xi_2$}
\psbrace[fillcolor=gray,rot=-90,ref=bC](6.3,3.2)(4.9,3.2){$\xi_3$}
\psbrace[fillcolor=gray,rot=-90,ref=bC](7.5,3.2)(6.3,3.2){$\xi_4$}
\psbrace[fillcolor=gray,rot=-90,ref=bC](9.9,3.2)(8.7,3.2){$\xi_4$}
\psbrace[fillcolor=gray,rot=-90,ref=bC](11.6,3.2)(9.9,3.2){$\xi_5$}

\psbrace[fillcolor=gray,rot=90,ref=tC](0,2.1)(1.3,2.1){$\mu_1$}
\psbrace[fillcolor=gray,rot=90,ref=tC](1.3,2.1)(2.5,2.1){$\mu_2$}
\psbrace[fillcolor=gray,rot=90,ref=tC](3.7,2.1)(4.9,2.1){$\mu_2$}
\psbrace[fillcolor=gray,rot=90,ref=tC](4.9,2.1)(6.3,2.1){$\mu_3$}
\psbrace[fillcolor=gray,rot=90,ref=tC](6.3,2.1)(7.5,2.1){$\mu_4$}
\psbrace[fillcolor=gray,rot=90,ref=tC](8.7,2.1)(9.9,2.1){$\mu_4$}
\psbrace[fillcolor=gray,rot=90,ref=tC](9.9,2.1)(11.6,2.1){$\mu_5$}

\end{pspicture}
\caption{Double stranded structure built with strings $r_j\in L_1$ and $s_j\in L_2$, in the case in which $L_1$ is a context-free language and $L_2$ is a regular language.}
\label{strand2}
\end{figure}

For $j$ large enough, repetitions of $y_s$ overlap with repetitions of $v_r$ (if $|v_r|>0$) and $y_r$ (if $|y_r|>0$).   
Then, we may seek expressions of $r_j$ and $s_j$ of the form 
\begin{align}
r_j &=\xi_1\xi_2^{j-j_0}\xi_3\xi_4^{j-j_0}\xi_5, \label{ri21}\\
s_j &=\mu_1\mu_2^{j-j_0}\mu_3\mu_4^{j-j_0}\mu_5, \label{si21}
\end{align}
with substrings $\xi_k$ and $\mu_k$ such that $|\xi_k|=|\mu_k|$ for $k=1, 2, \ldots, 5$, and for $j\ge j_0$, where $j_0$ is a constant to be determined.  

In Eqs.~(\ref{rj2}) and (\ref{sj2}), the iterated portions are decomposed as
\begin{align}
 v_r^{|y_s|j+1} & =\alpha_1\xi_2^{j-j_0}\alpha_2,\\
  y_r^{|y_s|j+1} & =\beta_1\xi_4^{j-j_0}\beta_2,\\
y_s^{|v_ry_r|j+1} & =\gamma_1\mu_2^{j-j_0}\gamma_2\mu_4^{j-j_0}\gamma_3,
\end{align}
with 
\begin{align}
\alpha_1\alpha_2  &= v_r^{j_0|y_s| + 1},\label{aux12}\\
\beta_1\beta_2 &= y_r^{j_0|y_s| + 1},\label{aux22}\\
 \gamma_1\gamma_2\gamma_3&= y_s^{j_0|v_ry_r| + 1},\label{aux32}
\end{align}
Consequently, $\xi_2$ and $\mu_2$ will be formed by repetitions of strings $v_r$ and $y_s$, respectively, with $|\xi_2|=|\mu_2|=|v_r||y_s|$. Also, $\xi_4$ and $\mu_4$ will be formed by repetitions of strings $y_r$ and $y_s$, respectively, with $|\xi_4|=|\mu_4|=|y_r||y_s|$.

Next, we set $\xi_1 =u_r\alpha_1$, $\xi_3 =\alpha_2x_r\beta_1$, $\xi_5 =\alpha_2z_r$, $\mu_1 =x_s\gamma_1$, $\mu_3=\gamma_2$, $\mu_5 =\gamma_3z_s$. 
Substrings $\alpha_1$, $\alpha_2$, $\beta_1$, $\beta_2$, $\gamma_1$, $\gamma_2$, and $\gamma_3$ are selected so that $|\xi_1=\mu_1|$, $|\xi_3=\mu_3|$, and $|\xi_5=\mu_5|$, with
\begin{equation}
\left\{
\begin{array}{l}
\gamma_1=\eee, |\alpha_1|= |x_s|-|u_r|,  \text{ if } |x_s|\ge|u_r|,\\
\alpha_1=\eee, |\gamma_1|= |u_r|-|x_s|,  \text{ if } |x_s|<|u_r|,  
\end{array}
\right.
\end{equation}
\begin{equation}
\left\{
\begin{array}{l}
\gamma_3=\eee, |\beta_2|= |z_s|-|z_r|,  \text{ if } |z_s|\ge|z_r|,\\
\beta_2=\eee, |\gamma_3|= |z_r|-|z_s|,  \text{ if } |z_s|<|z_r|,  
\end{array}
\right.
\end{equation}
and using Eqs.\ (\ref{aux12}) to (\ref{aux32}) for a value of $j_0$ large enough.

Finally, we compute $h(r_j, s_j)$ using Eqs.\ (\ref{ri21}), (\ref{si21}), and the identity in Eq.\ (\ref{prop1}). Then, letting $i=j-j_0$,
\begin{multline}
h(\xi_1\xi_2^i\xi_3\xi_4^i\xi_5, \mu_1\mu_2^i\mu_3\mu_4^i\mu_5)  =\\ (\xi_5)_\str{u}^\mathcal{R}[(\xi_4)_\str{u}^\mathcal{R}]^i(\xi_3)_\str{u}^\mathcal{R}[(\xi_2)_\str{u}^\mathcal{R}]^i(\xi_1)_\str{u}^\mathcal{R}(\xi_1)_\str{d}(\xi_2)_\str{d}^i(\xi_3)_\str{d}(\xi_4)_\str{d}^i(\xi_5)_\str{d}.
\label{h2}
\end{multline}

The condition stated by the lemma is obtained from Eq.~(\ref{h2}) with  $w_1=(\xi_5)_\str{u}^\mathcal{R}$, $w_2=(\xi_4)_\str{u}^\mathcal{R}$, $w_3=(\xi_3)_\str{u}^\mathcal{R}$, $w_4=(\xi_2)_\str{u}^\mathcal{R}$, $w_5=(\xi_1)_\str{u}^\mathcal{R}(\xi_1)_\str{d}$, $w_6=(\xi_2)_\str{d}$, $w_7=(\xi_3)_\str{d}$, $w_8=(\xi_4)_\str{d}$, $w_9=(\xi_5)_\str{d}$. Also, $|w_2w_4w_6w_8|=|\xi_2\xi_4|=|v_ry_r||y_s|>0$.

Consider next the case $L \in \mathcal{F}(\text{\sffamily  REG}, \text{\sffamily CF})$.
 This case is treated as the previous one,  except that $r$ and $s$ are decomposed following the pumping lemmas for regular languages and context-free languages, respectively. Thus, strings $r_j$ and $s_j$ are now defined as 
\begin{align}
r_j& = x_ry_r^{|v_sy_s|j+1}z_r,\label{rj3}\\
s_j& = u_sv_s^{|y_r|j+1}x_sy_s^{|y_r|j+1}z_s,\label{sj3}
\end{align}
with $|y_r|>0$,  $|v_sy_s|>0$, and for any $j\ge 0$.  

In Eqs.~(\ref{rj3}) and (\ref{sj3}), the iterated portions are decomposed as
\begin{align}
y_r^{|v_sy_r|s+1} & =\alpha_1\xi_2^{j-j_0}\alpha_2\xi_4^{j-j_0}\alpha_3,\\
 v_s^{|y_r|j+1} & =\beta_1\mu_2^{j-j_0}\beta_2,\\
  y_s^{|y_r|j+1} & =\gamma_1\mu_4^{j-j_0}\gamma_2,
\end{align}
with 
\begin{align}
 \alpha_1\alpha_2\alpha_3&= y_r^{j_0|v_sy_s| + 1},\\
\beta_1\beta_2  &= v_s^{j_0|y_r| + 1},\\
\gamma_1\gamma_2 &= y_s^{j_0|y_r| + 1}.
\end{align}

The demonstration then follows similar steps as in the previous case: $r_j$ and $s_j$ are expressed as in Eqs. (\ref{ri21}) and (\ref{si21}) with $\xi_1 =x_r\alpha_1$, $\xi_3 =\alpha_2$, $\xi_5 =\alpha_3z_r$, $\mu_1 =u_s\beta_1$, $\mu_3 =\beta_2x_s\gamma_1$,  $\mu_5 =\gamma_2z_s$. 

\end{proof}

The last lemma considers the case in which both the core and the folding procedure languages are context-free. It has a longer demonstration; nevertheless, it follows the same technique as the previous ones.

\begin{lemma}
Let $L \in \mathcal{F}(\text{\sffamily CF}, \text{\sffamily CF})$ be an infinite language over an alphabet $\Sigma$. Then, there are strings $w_1, w_2, \ldots, w_{13}\in \Sigma^*$, with  $|w_2w_4w_6\cdots w_{12}| >0$, such that \linebreak $w_1w_2^iw_3w_4^iw_5\cdots w_{11}w_{12}^iw_{13}\in L$ for all $i\ge 0$. 

\end{lemma}
\begin{proof}
We proceed as in the previous lemmas, except that the pumping lemma for context-free languages is used for both  strings $r\in L_1$ and $s\in L_2$. Thus, $r$ is written as $r=u_rv_rx_ry_rz_r$, with  $|v_ry_r|>0$ and $r_j=u_rv_r^jx_ry_r^jz_r\in L_1$ for any $j\ge 0$, and $s$ is written as $s=u_sv_sx_sy_sz_s$, with  $|v_sy_s|>0$ and $s_j=u_sv_s^jx_sy_s^jz_s\in L_2$ for any $j\ge0$. 

Consider first the case in which $|v_r||y_s|> |v_s||y_r|$. Strings $r_j$ and $s_j$ are defined as 
\begin{align}
r_j& = u_rv_r^{|v_r||y_s||v_sy_s|j+1}x_ry_r^{|v_r||y_s||v_sy_s|j+1}z_r,\label{ri3}\\
s_j& = u_sv_s^{|v_r||y_s||v_ry_r|j+1}x_sy_s^{|v_r||y_s||v_ry_r|j+1}z_s, \label{si3}
\end{align}
for any $j\ge 0$, and their double stranded arrangement is illustrated in Fig.~\ref{strand3}.   

\begin{figure}
\centering
\begin{pspicture}(-2,1)(12,4.5)
\uput[0](-1.75,2.67){
$\begin{bmatrix}
r_j\\
s_j
\end{bmatrix}
=
\begin{bmatrix}
\hspace*{11.8cm}\\
\ 
\end{bmatrix}
$
}

\psframe[linecolor=fixed, fillstyle=solid,fillcolor=fixed](0,2.7)(11.6,3.15)
\psframe[linecolor=iter1, fillstyle=solid,fillcolor=iter1](1,2.7)(7,3.15)
\psframe[linecolor=iter1, fillstyle=solid,fillcolor=iter1](7.6,2.7)(10.1,3.15)
\psframe[linecolor=iter2, fillstyle=solid,fillcolor=iter2](1.3,2.7)(3.5,3.15)
\psframe[linecolor=iter2, fillstyle=solid,fillcolor=iter2](5.1,2.7)(6.4,3.15)
\psframe[linecolor=iter2, fillstyle=solid,fillcolor=iter2](8.2,2.7)(10.1,3.15)

\psframe[linecolor=fixeda, fillstyle=solid,fillcolor=fixeda](0,2.15)(11.6,2.58)
\psframe[linecolor=iter1a, fillstyle=solid,fillcolor=iter1a](1.3,2.15)(4,2.58)
\psframe[linecolor=iter1a, fillstyle=solid,fillcolor=iter1a](4.6,2.15)(10.6,2.58)
\psframe[linecolor=iter2a, fillstyle=solid,fillcolor=iter2a](1.3,2.15)(3.5,2.58)
\psframe[linecolor=iter2a, fillstyle=solid,fillcolor=iter2a](5.1,2.15)(6.4,2.58)
\psframe[linecolor=iter2a, fillstyle=solid,fillcolor=iter2a](8.2,2.15)(10.1,2.58)

\uput[90](0,2.55){$|$}
\uput[90](0.5,2.55){$u_r$}
\uput[90](1,2.55){$|$}
\psline[linewidth=.5pt]{<-}(1,2.9)(2.7,2.9)
\psline[linewidth=.5pt]{->}(5.3,2.9)(7,2.9)
\uput[90](4,2.55){\small $v_r^{|v_r||y_s||v_sy_s|j+1}$}
\uput[90](7,2.55){$|$}
\uput[90](7.3,2.55){$x_r$}
\uput[90](7.6,2.55){$|$}
\uput[90](8.85,2.55){\small $y_r^{|v_r||y_s||v_sy_s|j+1}$}

\uput[90](10.1,2.55){$|$}
\uput[90](10.85,2.55){$z_r$}
\uput[90](11.6,2.55){$|$}

\psline[linewidth=.5pt](-1.28,2.65)(-.95,2.65)
\psline[linewidth=.5pt](0,2.65)(11.6,2.65)

\uput[90](0,2){$|$}
\uput[90](0.65,2){$u_s$}
\uput[90](1.3,2){$|$}
\uput[90](2.65,2){\small $v_s^{|v_r||y_s||v_ry_r|j+1}$}

\uput[90](4,2){$|$}
\uput[90](4.3,2){$x_s$}
\uput[90](4.6,2){$|$}

\psline[linewidth=.5pt]{<-}(4.6,2.35)(6.3,2.35)
\psline[linewidth=.5pt]{->}(8.9,2.35)(10.6,2.35)
\uput[90](7.6,2){\small $y_s^{|v_r||y_s||v_ry_r|j+1}$}

\uput[90](10.6,2){$|$}
\uput[90](11.1,2){$z_s$}
\uput[90](11.6,2){$|$}

\psbrace[fillcolor=gray,rot=-90,ref=bC](1.3,3.2)(0,3.2){$\xi_1$}
\psbrace[fillcolor=gray,rot=-90,ref=bC](3.5,3.2)(1.3,3.2){$\xi_2^{j-j_0}$}
\psbrace[fillcolor=gray,rot=-90,ref=bC](5.1,3.2)(3.5,3.2){$\xi_3$}
\psbrace[fillcolor=gray,rot=-90,ref=bC](6.4,3.2)(5.1,3.2){$\xi_4^{j-j_0}$}
\psbrace[fillcolor=gray,rot=-90,ref=bC](8.2,3.2)(6.4,3.2){$\xi_5$}
\psbrace[fillcolor=gray,rot=-90,ref=bC](10.1,3.2)(8.2,3.2){$\xi_6^{j-j_0}$}
\psbrace[fillcolor=gray,rot=-90,ref=bC](11.6,3.2)(10.1,3.2){$\xi_7$}

\psbrace[fillcolor=gray,rot=90,ref=tC](0,2.1)(1.3,2.1){$\mu_1$}
\psbrace[fillcolor=gray,rot=90,ref=tC](1.3,2.1)(3.5,2.1){$\mu_2^{j-j_0}$}
\psbrace[fillcolor=gray,rot=90,ref=tC](3.5,2.1)(5.1,2.1){$\mu_3$}
\psbrace[fillcolor=gray,rot=90,ref=tC](5.1,2.1)(6.4,2.1){$\mu_4^{j-j_0}$}
\psbrace[fillcolor=gray,rot=90,ref=tC](6.4,2.1)(8.2,2.1){$\mu_5$}
\psbrace[fillcolor=gray,rot=90,ref=tC](8.2,2.1)(10.1,2.1){$\mu_6^{j-j_0}$}
\psbrace[fillcolor=gray,rot=90,ref=tC](10.1,2.1)(11.6,2.1){$\mu_7$}

\end{pspicture}
	\caption{Double stranded structure built with strings $r_j\in L_1$ and $s_j\in L_2$, in the case in which both $L_1$ and $L_2$ are context-free languages and $|v_r||y_s|>|v_s||y_r|$.}
\label{strand3}
\end{figure}
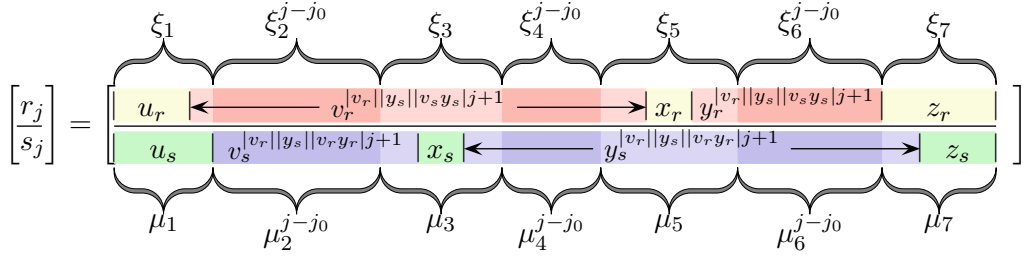

For $j$ large enough, and since $|v_r||y_s|> |v_s||y_r|$, then repetitions of $v_r$ overlap with repetitions of both $v_s$ (if $|v_s|>0$) and $y_s$, and repetitions of $y_s$ overlap with repetitions of both $v_r$ and $y_r$ (if $|y_r|>0$). Then, we may seek expressions of $r_j$ and $s_j$ of the form 
\begin{align}
r_j &=\xi_1\xi_2^{j-j_0}\xi_3\xi_4^{j-j_0}\xi_5\xi_6^{j-j_0}\xi_7, \label{ri41}\\
s_j &=\mu_1\mu_2^{j-j_0}\mu_3\mu_4^{j-j_0}\mu_5\mu_6^{j-j_0}\mu_7, \label{si41}
\end{align}
with substrings $\xi_k$ and $\mu_k$ such that $|\xi_k|=|\mu_k|$ for $k=1, 2, \ldots, 7$, and for $j\ge j_0$, where $j_0$ is a constant to be determined.

In Eqs.~(\ref{ri3}) and (\ref{si3}), the iterated portions may be decomposed as
\begin{align}
 v_r^{|v_r||y_s||v_sy_s|j+1} & = \alpha_1\xi_2^{(j-j_0)}\alpha_2\xi_4{(j-j_0)}\alpha_3, \\
y_r^{|v_r||y_s||v_sy_s|j+1} &=\beta_1\xi_6^{(j-j_0)}\beta_2,\\
v_s^{|v_r||y_s||v_ry_r|j+1} &=\gamma_1\mu_2^{(j-j_0)}\gamma_2,\\
y_s^{|v_r||y_s||v_ry_r|j+1} &=\delta_1\mu_4^{(j-j_0)}\delta_2\mu_6^{(j-j_0)}\delta_3,
\end{align}

with 
\begin{align}
\alpha_1\alpha_2\alpha_3  &= v_r^{j_0|v_r||y_s||v_sy_s| + 1},\label{aux13}\\
\beta_1\beta_2 &= y_r^{j_0|v_r||y_s||v_sy_s| + 1},\label{aux23}\\
 \gamma_1\gamma_2&= v_s^{j_0|v_r||y_s||v_ry_r| + 1}\label{aux33}\\
 \delta_1\delta_2\delta_3&= v_s^{j_0|v_r||y_s||v_ry_r| + 1}.\label{aux43}
\end{align}
Consequently, $\xi_2$ and $\mu_2$ will be formed by repetitions of strings $v_r$ and $v_s$, respectively, with $|\xi_2|=|\mu_2|=|v_r||v_s||y_s||v_ry_r|$; $\xi_4$ and $\mu_4$ will be formed by repetitions of strings $v_r$ and $y_s$, respectively, with $|\xi_4|=|\mu_4|=|v_r||y_s|(|v_r||y_s|-|v_s||y_r|)$; and $\xi_6$ and $\mu_6$ will be formed by repetitions of strings $y_r$ and $y_s$, respectively, with $|\xi_6|=|\mu_6|=|v_r||y_r||y_s||v_sy_s|$.
Next, we set $\xi_1 =u_r\alpha_1$, $\xi_3 =\alpha_2$, $\xi_5=\alpha_3x_r\beta_1$, $\xi_7 =\beta_2z_r$, $\mu_1 =u_s\gamma_1$, $\mu_3=\gamma_2x_s\delta_1$, $\mu_5 =\delta_2$, $\mu_7=\delta_3z_s$. 

Substrings $\alpha_1$, $\alpha_2$, $\alpha_3$, $\beta_1$, $\beta_2$, $\gamma_1$, $\gamma_2$, $\delta_1$, $\delta_2$ and $\delta_3$ are selected so that $|\xi_1=\mu_1|$, $|\xi_3=\mu_3|$, $|\xi_5=\mu_5|$, $|\xi_7=\mu_7|$, with $\delta_1=\eee$, $|\alpha_2|=|\gamma_2| +|x_s|$, 
\begin{equation}
\left\{
\begin{array}{l}
\gamma_1=\eee, |\alpha_1|= |u_s|-|u_r|,  \text{ if } |u_s|\ge|u_r|,\\
\alpha_1=\eee, |\gamma_1|= |u_r|-|u_s|,  \text{ if } |u_s|<|u_r|, 
\end{array}
\right.
\label{aaa}
\end{equation}
\begin{equation}
\left\{
\begin{array}{l}
\delta_3=\eee, |\beta_2|= |z_s|-|z_r|,  \text{ if } |z_s|\ge|z_r|,\\
\beta_2=\eee, |\delta_3|= |z_r|-|z_s|,  \text{ if } |z_s|<|z_r|,  
\end{array}
\right.
\end{equation}
and using Eqs.\ (\ref{aux13}) to (\ref{aux43}) for a value of $j_0$  large enough.

Finally, we compute $h(r_j, s_j)$ using Eqs.\ (\ref{ri41}), (\ref{si41}), and the identity in Eq.\ (\ref{prop1}). Then, letting $i=j-j_0$,
\begin{multline}
h(\xi_1\xi_2^i\xi_3\xi_4^i\xi_5\xi_6^i\xi_7, \mu_1\mu_2^i\mu_3\mu_4^i\mu_5\mu_6^i\mu_7)  =\\(\xi_7)_\str{u}^\mathcal{R}[(\xi_6)_\str{u}^\mathcal{R}]^i (\xi_5)_\str{u}^\mathcal{R}[(\xi_4)_\str{u}^\mathcal{R}]^i(\xi_3)_\str{u}^\mathcal{R}[(\xi_2)_\str{u}^\mathcal{R}]^i(\xi_1)_\str{u}^\mathcal{R}(\xi_1)_\str{d}(\xi_2)_\str{d}^i(\xi_3)_\str{d}(\xi_4)_\str{d}^i(\xi_5)_\str{d}(\xi_6)_\str{d}^i(\xi_7)_\str{d}.
\label{h3}
\end{multline}

The condition stated by the lemma is obtained from Eq.~(\ref{h3}) with  $w_1=(\xi_7)_\str{u}^\mathcal{R}$, $w_2=(\xi_6)_\str{u}^\mathcal{R}$, $w_3=(\xi_5)_\str{u}^\mathcal{R}$, $w_4=(\xi_4)_\str{u}^\mathcal{R}$, $w_5=(\xi_3)_\str{u}^\mathcal{R}$, $w_6=(\xi_2)_\str{u}^\mathcal{R}$, $w_7=(\xi_1)_\str{u}^\mathcal{R}(\xi_1)_\str{d}$, $w_8=(\xi_2)_\str{d}$, $w_9=(\xi_3)_\str{d}$, $w_{10}=(\xi_4)_\str{d}$, $w_{11}=(\xi_5)_\str{d}$, $w_{12}=(\xi_6)_\str{d}$, $w_{13}=(\xi_7)_\str{d}$. Also, $|w_2w_4w_6w_8w_{10}w_{12}|=|\xi_2\xi_4\xi_6|=|v_r||y_s||v_ry_r||v_sy_s|> 0$.

Consider next the case in which $|v_r||y_s|<|v_s||y_r|$, and define
\begin{align}
r_j& = u_rv_r^{|v_s||y_r||v_sy_s|j+1}x_ry_r^{|v_s||y_r||v_sy_s|j+1}z_r,\label{ri3b}\\
s_j& = u_sv_s^{|v_s||y_r||v_ry_r|j+1}x_sy_s^{|v_s||y_r||v_ry_r|j+1}z_s,\label{si3b}
\end{align}

for any $j\ge 0$, with the double stranded arrangement  illustrated in Fig.~\ref{strand4}.   

\begin{figure}
\centering
\begin{pspicture}(-2,1)(12,4.5)
\uput[0](-1.75,2.67){
$\begin{bmatrix}
r_j\\
s_j
\end{bmatrix}
=
\begin{bmatrix}
\hspace*{11.8cm}\\
\ 
\end{bmatrix}
$
}

\psframe[linecolor=fixed, fillstyle=solid,fillcolor=fixed](0,2.15)(11.6,2.58)
\psframe[linecolor=iter1, fillstyle=solid,fillcolor=iter1](1,2.15)(7,2.58)
\psframe[linecolor=iter1, fillstyle=solid,fillcolor=iter1](7.6,2.15)(10.1,2.58)
\psframe[linecolor=iter2, fillstyle=solid,fillcolor=iter2](1.3,2.15)(3.5,2.58)
\psframe[linecolor=iter2, fillstyle=solid,fillcolor=iter2](5.1,2.15)(6.4,2.58)
\psframe[linecolor=iter2, fillstyle=solid,fillcolor=iter2](8.2,2.15)(10.1,2.58)

\psframe[linecolor=fixeda, fillstyle=solid,fillcolor=fixeda](0,2.7)(11.6,3.15)
\psframe[linecolor=iter1a, fillstyle=solid,fillcolor=iter1a](1.3,2.7)(4,3.15)
\psframe[linecolor=iter1a, fillstyle=solid,fillcolor=iter1a](4.6,2.7)(10.6,3.15)
\psframe[linecolor=iter2a, fillstyle=solid,fillcolor=iter2a](1.3,2.7)(3.5,3.15)
\psframe[linecolor=iter2a, fillstyle=solid,fillcolor=iter2a](5.1,2.7)(6.4,3.15)
\psframe[linecolor=iter2a, fillstyle=solid,fillcolor=iter2a](8.2,2.7)(10.1,3.15)

\uput[90](0,2.){$|$}
\uput[90](0.5,2.){$u_s$}
\uput[90](1,2.){$|$}
\psline[linewidth=.5pt]{<-}(1,2.35)(2.7,2.35)
\psline[linewidth=.5pt]{->}(5.3,2.35)(7,2.35)
\uput[90](4,2.){\small $v_s^{|v_s||y_r||v_ry_r|j+1}$}
\uput[90](7,2.){$|$}
\uput[90](7.3,2.){$x_s$}
\uput[90](7.6,2.){$|$}
\uput[90](8.85,2.){\small $y_s^{|v_s||y_r||v_ry_r|j+1}$}

\uput[90](10.1,2.){$|$}
\uput[90](10.85,2.){$z_s$}
\uput[90](11.6,2.){$|$}

\psline[linewidth=.5pt](-1.28,2.65)(-.95,2.65)
\psline[linewidth=.5pt](0,2.65)(11.6,2.65)

\uput[90](0,2.55){$|$}
\uput[90](0.65,2.55){$u_r$}
\uput[90](1.3,2.55){$|$}
\uput[90](2.65,2.55){\small $v_r^{|v_s||y_r||v_sy_s|j+1}$}

\uput[90](4,2.55){$|$}
\uput[90](4.3,2.55){$x_r$}
\uput[90](4.6,2.55){$|$}

\psline[linewidth=.5pt]{<-}(4.6,2.9)(6.3,2.9)
\psline[linewidth=.5pt]{->}(8.9,2.9)(10.6,2.9)

\uput[90](7.6,2.55){\small $y_r^{|v_s||y_r||v_sy_s|j+1}$}

\uput[90](10.6,2.55){$|$}
\uput[90](11.1,2.55){$z_r$}
\uput[90](11.6,2.55){$|$}

\psbrace[fillcolor=gray,rot=-90,ref=bC](1.3,3.2)(0,3.2){$\xi_1$}
\psbrace[fillcolor=gray,rot=-90,ref=bC](3.5,3.2)(1.3,3.2){$\xi_2^{j-j_0}$}
\psbrace[fillcolor=gray,rot=-90,ref=bC](5.1,3.2)(3.5,3.2){$\xi_3$}
\psbrace[fillcolor=gray,rot=-90,ref=bC](6.4,3.2)(5.1,3.2){$\xi_4^{j-j_0}$}
\psbrace[fillcolor=gray,rot=-90,ref=bC](8.2,3.2)(6.4,3.2){$\xi_5$}
\psbrace[fillcolor=gray,rot=-90,ref=bC](10.1,3.2)(8.2,3.2){$\xi_6^{j-j_0}$}
\psbrace[fillcolor=gray,rot=-90,ref=bC](11.6,3.2)(10.1,3.2){$\xi_7$}

\psbrace[fillcolor=gray,rot=90,ref=tC](0,2.1)(1.3,2.1){$\mu_1$}
\psbrace[fillcolor=gray,rot=90,ref=tC](1.3,2.1)(3.5,2.1){$\mu_2^{j-j_0}$}
\psbrace[fillcolor=gray,rot=90,ref=tC](3.5,2.1)(5.1,2.1){$\mu_3$}
\psbrace[fillcolor=gray,rot=90,ref=tC](5.1,2.1)(6.4,2.1){$\mu_4^{j-j_0}$}
\psbrace[fillcolor=gray,rot=90,ref=tC](6.4,2.1)(8.2,2.1){$\mu_5$}
\psbrace[fillcolor=gray,rot=90,ref=tC](8.2,2.1)(10.1,2.1){$\mu_6^{j-j_0}$}
\psbrace[fillcolor=gray,rot=90,ref=tC](10.1,2.1)(11.6,2.1){$\mu_7$}

\end{pspicture}
\caption{Double stranded structure built with strings $r_j\in L_1$ and $s_j\in L_2$, in the case in which both $L_1$ and $L_2$ are context-free languages and $|v_r||y_s|<|v_s||y_r|$.}
\label{strand4}
\end{figure}
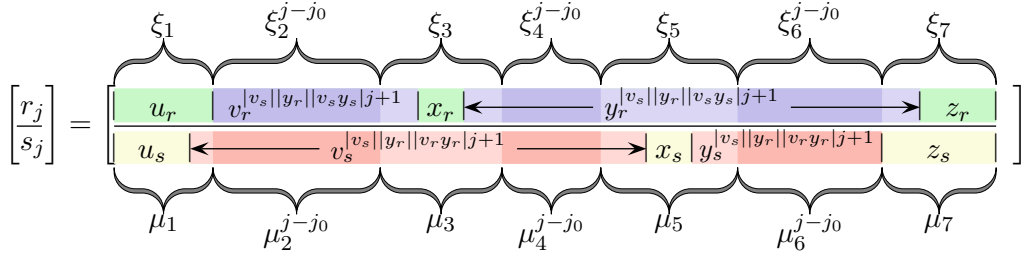

We seek expressions of $r_j$ and $s_j$ of the form in Eqs.\ (\ref{ri41}) and (\ref{si41}), and decompose the iterated portions in Eqs.~(\ref{ri3b}) and (\ref{si3b}) as 
\begin{align}
v_r^{|v_s||y_r||v_sy_s|j+1} & = \alpha_1\xi_2^{(j-j_0)}\alpha_2, \\
y_r^{|v_s||y_r||v_sy_s|j+1} &=\beta_1\xi_4^{(j-j_0)}\beta_2\xi_6^{(j-j_0)}\beta_3,\\
v_s^{|v_s||y_r||v_ry_r|j+1} &=\gamma_1\mu_2^{(j-j_0)}\gamma_2\mu_4^{(j-j_0)}\gamma_3,\\
y_s^{|v_s||y_r||v_ry_r|j+1} &=\delta_1\mu_6^{(j-j_0)}\delta_2,
\end{align}

with 
\begin{align}
\alpha_1\alpha_2  &= v_r^{j_0|v_s||y_r||v_sy_s| + 1},\label{aux14}\\
\beta_1\beta_2\beta_3 &= y_r^{j_0|v_s||y_r||v_sy_s| + 1},\label{aux24}\\
 \gamma_1\gamma_2\gamma_3&= v_s^{j_0|v_s||y_r||v_ry_r| + 1},\label{aux34}\\
 \delta_1\delta_2&= v_s^{j_0|v_s||y_r||v_ry_r| + 1}.\label{aux44}
\end{align}

Consequently, $\xi_2$ and $\mu_2$ will be formed by repetitions of strings $v_r$ and $v_s$, respectively, with $|\xi_2|=|\mu_2|=|v_r||v_s||y_r||v_sy_s|$; $\xi_4$ and $\mu_4$ will be formed by repetitions of strings $y_r$ and $v_s$, respectively, with $|\xi_4|=|\mu_4|=|y_r||v_s|(|v_s||y_r|-|v_r||y_s|)$; and $\xi_6$ and $\mu_6$ will be formed by repetitions of strings $y_r$ and $y_s$, respectively, with $|\xi_6|=|\mu_6|=|y_s||v_s||y_r||v_ry_r|$.
Next, we set $\xi_1 =u_r\alpha_1$, $\xi_3 =\alpha_2x_r\beta_1$, $\xi_5=\beta_2$, $\xi_7 =\beta_3z_r$, $\mu_1 =u_s\gamma_1$, $\mu_3=\gamma_2$, $\mu_5 =\gamma_3x_s\delta_1$, $\mu_7=\delta_2z_s$. 

Substrings $\alpha_1$, $\alpha_2$, $\alpha_3$, $\beta_1$, $\beta_2$, $\gamma_1$, $\gamma_2$, $\delta_1$, $\delta_2$ and $\delta_3$ are selected so that $|\xi_1=\mu_1|$, $|\xi_3=\mu_3|$, $|\xi_5=\mu_5|$, $|\xi_7=\mu_7|$, with $\gamma_3=\eee$, $|\beta_2|=|\delta_1| +|x_s|$, 
\begin{equation}
\left\{
\begin{array}{l}
\delta_2=\eee, |\beta_3|= |z_s|-|z_r|,  \text{ if } |z_s|\ge|z_r|,\\
\beta_3=\eee, |\delta_2|= |z_r|-|z_s|,  \text{ if } |z_s|<|z_r|,  
\end{array}
\right.
\end{equation}
and using Eq.\ (\ref{aaa}) and Eqs.\ (\ref{aux14}) to (\ref{aux44}), for a value of $j_0$ large enough.

The demonstration then follows similar steps as in the previous case.

Finally, consider the remaining case of $|v_r||y_s|=|v_s||y_r|$. If $|v_r||y_s|\neq 0$ and $|v_s||y_r|\neq 0$, then the same constructions of the previous cases work, resulting in $\xi_4=\mu_4=\eee$ and therefore $w_4=w_{10}=\eee$.

If $|v_r||y_s|=|v_s||y_r|=0$, then either $v_r=v_s=\eee$ or $y_r=y_s=\eee$ (recall that the pumping lemma for context-free languages demands $|v_ry_r|>0$ and $|v_sy_s|>0$). 
Assume first $v_r=v_s=\eee$, and define
\begin{align}
r_j& = u_rx_ry_r^{|y_s|j+1}z _r,\\
s_j& = u_sx_sy_s^{|y_r|j+1}z_s,
\end{align}
for any $j\ge 0$. Next, follow the same procedure as in Lemma 1 to express $r_j$ and $s_j$ in the form of Eqs. (\ref{ri1}) and (\ref{si1}), 
with the exception that $\xi_1 =u_rx_r\alpha_1$ and  $\mu_1 =u_sx_s\beta_1$ (instead of $\xi_1 =x_r\alpha_1$ and  $\mu_1 =x_s\beta_1$, respectively). Then, $h(r_j, s_j)$ is given by Eq.\ (\ref{h1}) and the condition stated by the present lemma is obtained with  $w_1=(\xi_3)_\str{u}^\mathcal{R}$, $w_2=(\xi_2)_\str{u}^\mathcal{R}$, $w_3=w_4=w_5=w_6=\eee$, $w_7=(\xi_1)_\str{u}^\mathcal{R}(\xi_1)_\str{d}$, $w_8=w_9=w_{10}=w_{11}=\eee$, $w_{12}=(\xi_2)_\str{d}$, $w_{13}=(\xi_3)_\str{d}$. Also, $|w_2w_4w_6w_8w_{10}w_{12}|=|\xi_2|=|y_r||y_s|\ge 1$.
 
If $y_r=y_s=\eee$, proceed as above  with 
\begin{align}
r_j& = u_rv_r^{|v_s|j+1}x_rz _r,\\
s_j& = u_sv_s^{|v_r|j+1}x_sz_s,
\end{align}
for any $j\ge 0$.
\end{proof}

\section{Final remarks and example}

The lemmas only apply to infinite languages. However, any finite language is in class $\mathcal{F}(\text{\sffamily  REG}, \text{\sffamily REG})$. Let $L_1$ be any arbitrary language and define $\Phi=(L_1, \str{d}^*)$; then, $L_1=L(\Phi)$. Therefore, if $L_1$ finite then it is regular and, consequently,  $L_1\in\mathcal{F}(\text{\sffamily  REG}, \text{\sffamily REG})$. Further, and since $\text{\sffamily  REG}\subset \text{\sffamily CF}$, then $L_1$ also is in 
$\mathcal{F}(\text{\sffamily  CF}, \text{\sffamily REG})$, $\mathcal{F}(\text{\sffamily  REG}, \text{\sffamily CF})$, and \linebreak $\mathcal{F}(\text{\sffamily  CF}, \text{\sffamily CF})$.

In spite of the lemmas being weak, they still are useful to prove non membership of some languages  in a class, as the following example shows.

\begin{example}
Consider $L=\{\str{a}^n|\,n$ is prime $\}$ and assume that $L$ satisfies the lemma. Then, there are strings $w_1, w_2, \ldots, w_{13}\in \str{a}^*$, with  $|w_2w_4w_6\cdots w_{12}| >0$, such that $w_1w_2^iw_3w_4^iw_5\cdots w_{11}w_{12}^iw_{13}\in L$ for all $i\ge 0$. Letting $i=1$, we have that \linebreak  $w_1w_2w_3w_4w_5\cdots w_{11}w_{12}w_{13}=\str{a}^k$ for some prime $k$. Then, $w_2w_4w_6\cdots w_{12}=\str{a}^\ell$ with $0<\ell\le k$. Next, letting $i=k+1$, we obtain $w_1w_2^iw_3w_4^iw_5\cdots w_{11}w_{12}^iw_{13}=  \str{a}^{k(1+\ell)}$. However, this string is not in $L_2$, because the number of \str{a}'s is not a prime. The contradiction implies that  $L_2\notin\mathcal{F}(\text{\sffamily  CF}, \text{\sffamily CF})$.
\end{example}

\section*{Acknowledgements}
This work was supported by Conselho Nacional de Desenvolvimento Cient\'{i}fico e Tecnol\'{o}gico  (CNPq, Brazil).

\end{document}